\documentclass[sigconf]{acmart}
\AtBeginDocument{%
  }
\usepackage{latexsym}
\usepackage{amsfonts,amsmath}
\usepackage{array,supertabular}
\usepackage{anyfontsize}
\usepackage{pgfplots}
\usepackage{rotating}
\pgfplotsset{compat=1.15}
\usepackage{mathrsfs}
\usetikzlibrary{arrows}

\usepackage{stmaryrd}
\SetSymbolFont{stmry}{bold}{U}{stmry}{m}{n}
\usepackage{alltt}
\usepackage{epic,epsfig}
\usepackage{graphics}
\usepackage{algorithm}
\usepackage{algpseudocode}
\usepackage{algpseudocodex}
\usepackage{algorithm}
\usepackage{xcolor}
\usepackage{float}
\usepackage{tikz}
\usepackage{stackrel}
\usetikzlibrary{arrows.meta, positioning}
\usepackage{comment}
\usepackage{amsmath}

\setcopyright{acmlicensed}
\copyrightyear{2026}
\acmYear{2026}
\acmDOI{XXXXXXX.XXXXXXX}
\acmConference[Logic In Computer Science 2026]{LICS 2026}{June 03--05,
  2018}{Lisbon, Portugal}
\acmISBN{978-1-4503-XXXX-X/2018/06}

\newcommand{\lgx}[1]{| \hspace{-0.18mm}{#1} \hspace{-0.18mm}|}
\newcommand{\spc}{sp}
\mathchardef\mhyphen="2D
\newcommand{\card}{\mathrm{Card}}
\newcommand{\powset}{\mathcal P}

\newcommand{\K}{\ensuremath{\mathbf{K}}}
\newcommand{\nat}{\mathbb{N}}

\newcommand{\pipe}{\hspace{-0.21mm}|\hspace{-0.21mm}}

\newcommand{\inc}{\subseteq}

\newcommand{\Wpp}{\ensuremath{\tilde{W}}}

\newcommand{\chooseCCS}{\mbox{\texttt{ChCCS}}}
\newcommand{\sat}{\mbox{\texttt{Sat}}}

\newcommand{\CCS}{\mbox{\texttt{CCS}}}
\newcommand{\SF}{\mbox{\texttt{SF}}}
\newcommand{\CSF}{\mbox{\texttt{CSF}}}

\newcommand{\true}{\mbox{\texttt{True}}}

\newcommand{\all}{\mbox{\texttt{all}}}
\newcommand{\algand}{\mbox{\texttt{and}}}

\def\PSPACE{\mathbf{PSPACE}}
\def\card{\mathtt{Card}}

\def\NP{\mathbf{NP}}
\def\EXPSPACE{\mathbf{EXPSPACE}}

\def\NEXPTIME{\mathbf{NEXPTIME}}

\def\K{\mathbf{K}}

\def\Log2{\mathbf{\K De(2)}}
\def\At{\mathbf{At}}
\def\Fo{\mathbf{Fo}}

\newcommand{\lgu}{| \hspace{-0.18mm}u \hspace{-0.18mm}|}

\mathchardef\mhyphen="2D
\newcommand{\cal}{\mathcal}

\newcommand{\Bom}{\ensuremath{\Box^{\,\mhyphen}}}
\newcommand{\Ri}[1]{\ensuremath{R_{#1}}}
\renewcommand{\Log}[1]{\mathbf{\K De(#1)}}
\newcommand{\satW}{\mbox{\texttt{SatW}}}
\newcommand{\chooseW}{\mbox{\texttt{ChW}}}
\newcommand{\nextW}{\mbox{\texttt{NextW}}}

\newcommand{\voidsequence}{(\,)}

\newcommand{\sequencef}[3]{(#1_i)_{i\in[#2:#3]}}

\newcommand{\interf}[2]{[#1\!:\!#2]}
\newcommand{\intero}[2]{[#1\!:\!#2[}
\newcommand{\pair}[2]{\langle #1,#2\rangle}

\newcommand{\voidpair}{\pair{\voidsequence}{\voidsequence}}

\newcommand{\win}[3]{\langle (#1_i)_{i\in[0:#3]}, (#2_i)_{i\in[0:#3[}\rangle } 

\newcommand{\infwin}[2]{\langle (#1_i)_{0\leq i}, (#2_i)_{0\leq i}\rangle }
\newcommand{\partialwin}[4]{\langle (#2_i)_{i\in[#1:#4]}, (#3_i)_{i\in[#1:#4[}\rangle } 
\newcommand{\shrinkwin}[4]{\langle (#2_i)_{i\in[#1:#4]}, (#3_i[0])_{i\in[#1:#4[}\rangle }
\newcommand{\longwin}[4]{\langle (#1_i)_{i\in[0:#3]}, (#2_i)_{i\in[0:#3[}\rangle } 

\newcommand{\lbr}{[}
\newcommand{\rbr}{]}

\newcommand{\sizew}{n}
\newcommand{\dotminus}{{\stackbin{.}{-}}}
\newcommand{\uplim}[1]{\ensuremath{\chi(#1)}}

\begin{document}

\title{Parameterized complexity of \emph{n}-dense modal logics}

\author{Olivier Gasquet}
\email{olivier.gasquet@irit.fr}
\orcid{0009-0007-7083-9767}
\affiliation{%
  \institution{University of Toulouse - IRIT - CNRS}
  \city{Toulouse}
  \country{France}
}

\author{}
\email{}
\orcid{}
\affiliation{%
  \institution{}
  \city{}
  \country{}
}

\begin{abstract}
  Exact tight bounds of the complexity of the satisfiability problem for dense modal logics is a difficult question, likely somewhere between $\PSPACE$ and $\EXPSPACE$ depending of the logic under question. For a class of them, called here $n$-dense logics (characterized by axioms $\Box^n p\rightarrow \Box p$), we refine the known results --membership in $\NEXPTIME$-- in the light of parameterized complexity, as introduced in \cite{Downey}, and prove that they belong to the parameterized class para-$\PSPACE$: there exists a poly-space algorithm once the modal depth of the input is considered as a parameter. This is done by generalizing the novel analysis tool introduced in \cite{BalGasq25}, and therein called windows, to \emph{recursive windows}. 
\end{abstract}



\received{20 February 2007}
\received[revised]{12 March 2009}
\received[accepted]{5 June 2009}

\maketitle

\section{Introduction}

Modal logic constitutes a broad and well-established area of mathematical logic, with significant applications across a wide range of disciplines, from reasoning on programs to deontic aspect of reasoning. For instance, temporal logics play a central role in program verification, dynamic and epistemic logics are widely used in the analysis of multi-agents systems. Modal logics are commonly obtained by extending classical propositional logic with modal operators often denoted by $\Box$ (or $\Box_a$) which, according to the case can be interpreted as \emph{always}, \emph{ought to}, \emph{$a$ knows}, \emph{after action $a$}, etc. In all these uses, the question of determining the existence of a model for some logic $L$, the so-called $L$-sat problem, is crucial and constitutes a complete field of research, as well-as the design of algorithms that solve it. The complexity of these problems ranges from polynomial up to undecidability according to the logic, or the fragment of logic, considered. Nevertheless, most of the well-know modal logics have a $\PSPACE$-complete $L$-sat problem. 

Models for these logics are mainly based on so-called \emph{Kripke frame} which are graphs $(W,R)$ where nodes of $W$ are labeled with a valuation over a set of propositional variables, and, according to the logic $L$ in question, these graph must have additional properties (e.g.\ transitivity for the well-know logic $\K 4$, reflexivity, etc.). 

Among all modal logics, those whose frames have properties involving intermediary nodes (e.g.\ $(x,y)\in R$ implies $\exists z:(x,z)\in R$ and $(z,y)\in R$) have received less attention and the complexity of their sat-problem is not as well investigated as others, thought there is a general result of \cite{Lyon24} which establishes a quite high lower bound, namely membership in $\EXPSPACE$. 

In this work, we concentrate on a class of normal modal logics, referred to as $n$-dense modal logics and are characterized by axioms of the form $\Box^n p\rightarrow \Box p$ ($n>1$) and frames satisfying the property: $R\subseteq R^n$ (i.e.\ the existence of an edge between two nodes implies that of an $n$-long path between them). Their sat-problem can easily be proved to be in $\NEXPTIME$ by means of the filtration technique which goes back to \cite{Gabbay1972} (though it seems to originate in a paper \cite{LemmonScott66} which remained unpublished for many years). On another hand, we prove in section \ref{pspace-hardness} that they are $\PSPACE$-hard in \cite{BalGasq25}. Then, we will not give tight  bounds --it will remain a difficult open problem-- but we will make a step to finer characterize its complexity. Indeed, in the complexity approach of problems, it can sometimes been taken into account that if some parameter is fixed, then the problem becomes easier. So-called \emph{parameterized complexity theory} were mainly introduced by Downey and Fellows in \cite{Downey} and provides a framework for a refined analysis of hard algorithmic problems which are in $\NP$. It has lead to the identification of the Fixed-Parameter Tractable class (or para-$P$): if some parameter is fixed, then the problem becomes essentially polynomially hard. This is the case for Boolean satisfiability (the parameter being the number of variables) or the Vertex Cover (the number of vertices). Parametrized-complexity has later been extended by Flume \& Grohe in \cite{FLUM2003291} for other complexity classes, like para-$\PSPACE$: \emph{a problem $(Q, \kappa)$ belongs to the class para-$\PSPACE$ if there is a computable function $f:\nat\rightarrow\nat$, a polynomial $p$ and an algorithm that, given an input $x$ of size $\lgx{x}$, decides if $x\in Q$ in space $f(\kappa(x)).p(\lgx{x})$}. In the sequel, we will prove that there exists such one parameter for our problem, namely the so-called modal depth of the formulas (Note that this fact is not provable from the filtration technique, as it makes use of a model that contains exponentially many nodes independently from the modal depth of formulas). This limitation may be seen as acceptable in the sense that, for concrete applications, the modal depth is usually small and limited. 

In this paper, we design an algorithm, a tableau calculus more precisely, for the $2$-dense modal logic ($\K+\Box\Box p\rightarrow\Box p$). Semantical tableaux, or simply tableaux, are known as a powerful way of establishing complexity results for a wide range of logics by
means of decision procedures consisting in (un)successful attempts to proving the existence of a model. This algorithm will make use of a new technique, called \emph{windows}, that was introduced by Balbiani and Gasquet in \cite{BalGasq25} and in \cite{IGPL-Gasquet25} and which are small (polynomial) part of a tableau. We propose a consequent generalization of their work to \emph{recursive windows} and this will allow us to prove that the $n$-dense satisfiability problem is in para-$\PSPACE$ (the parameter being the modal depth of the input formulas), then we will argue for its extension to the general case of $n$-dense logics. The choice of focusing on the $2$-dense case is essentially motivated by the sake of clarity. 

After some definitions in the next section, we will prove $\PSPACE$-hardness of $n$-dense logic in section \ref{pspace-hardness}, then in section{tableaux} we will briefly review the basics of tableaux and their completeness proofs. We will go on in section \ref{simple_windows} by informally presenting the notion of window and will formally present it in section \ref{windows}. Sections \ref{algorithm} and \ref{analysis} will be devoted to the algorithm we designed and to its analysis (soundness, completeness and complexity) so to conclude for $2$-dense satisfiability. Finally, in section \ref{general_case} we will discuss the extension to all $n$-dense logics. 
\section{Basic definitions and settings}
\paragraph{Syntactical aspects} In the sequel, we will frequently identify finite sets of formulas with the conjunction of them: $s$ will stand for $\bigwedge_{\phi\in s} \phi$.

The language $\Fo$ of our logic is defined by $\Phi ::= \bot\pipe p\pipe \neg\Phi\pipe(\Phi\wedge\Phi)\pipe\Box\Phi$ where $p$ belongs to a given set $\At$ of propositional variables. As usual, $\Phi\vee\Psi$ abbreviates $\neg(\Phi\wedge\Psi)$ and $\Diamond\Phi$ abbreviates $\neg\Box\neg\Phi$. As usual too, $d(\phi)$ will denote the modal degree (or depth) of $\phi$ and $\lgx{\phi}$ its length\footnote{$d(p)=d(\bot)=0,d(\phi\wedge \psi)=\max\{d(\phi),d(\psi)\}, d(\neg\phi)=d(\phi),d(\Box\phi)=1+d(\phi)$ and $\lgx{p}=\lgx{\bot}=1, \lgx{\phi\wedge \psi}=\lgx{\phi}+\lgx{\psi}, \lgx{\neg\phi}=1+\lgx{\phi},\lgx{\Box\phi}=1+\lgx{\phi}$.}, both extend to sets by $\max$ and $+$ respectively. 

Our $n$-dense logics, henceforth denoted by $\Log{n}$, can be defined as follows: $\phi\in\Log{n}$ iff $\vdash\phi$ with
\begin{itemize}
    \item for all propositional tautologies $\phi$: $\vdash \phi$
    \item  Necessitation rule: if $\vdash \phi$ then $\vdash \Box\phi$
    \item Axiom K: for all $\phi,\psi\in\Fo\colon \vdash \Box(\phi\rightarrow\psi)\rightarrow\Box\phi\rightarrow\Box\psi$
    \item Axiom of $n$-density: for all $\phi\in\Fo\colon \vdash\Box^n\phi\rightarrow\Box\phi$
\end{itemize}
A formula $\phi\in\Fo$ is said to be $\Log{n}$-consistent iff $\neg\phi\not\in\Log{n}$. \\

\noindent The set $\SF(\Phi)$ of \emph{subformulas} of $\Phi$ is defined as usual and naturally extends to $\SF(s)$ where $s$ is a set of formulas: $\SF(s)=\bigcup_{\phi\in s}\SF(\phi)$. The set $\CSF(s)$ of \emph{classical subformulas} has the same definition as $\SF$ but where modal subformulas are considered as propositional variables (thus $\CSF(p\wedge \neg\Box q)=\{p,\neg\Box q, \Box q\}$). A finite set $u$ of formulas is a \emph{Classically Consistent Saturation} ($\CCS$) (denoted by $u\in\CCS$) iff  for all formulas $\phi,\psi$:
\begin{itemize}
\item if $\phi \wedge \psi\in u$ then $\phi \in u$ and $\psi\in u$,
\item if $\neg (\phi \wedge \psi)\in u$ then $\neg \phi\in u$ or $\neg \psi\in u$,
\item if $\neg \neg \phi\in u$ then $\phi \in u$,
\item $\bot\not\in u$,
\item if $\neg \phi\in u$ then $\phi\not\in u$.
\end{itemize}
In addition, given a finite set $s$ of formulas, we write $u\in\CCS(s)$ iff $u\in\CCS$ and $s\subseteq u$. Note that for such $u$: $\lgx{u}$ is linear in $\lgx{s}$. \\
The reader will have notice that $\CCS(s)$ are in fact a set-theoretic version of a classical Disjunctive Normal Form, and as such verifies $s\leftrightarrow \bigvee_{u\in\CCS(s)}u$. 


\paragraph{Sets and relations} A binary relation $R$ over a set $W$ is a subset of $W^2$. The composition of relations is: $R\circ S\triangleq \{(x,z)\colon \exists y$ s.th. $(x,y)\in R$ and $(y,z)\in R\}$, exponentiation for $n>0$: $R^n\triangleq$ if $n=1$ then $R$ else $R\circ R^{n-1}$, and inverse: $R^{-n} \triangleq \{(y,x)\colon (x,y)\in R^n\}$. The powerset of a set $X$ is denoted by $\powset(X)$. 

\paragraph{Kripke semantics} A Kripke-model is a triple $(W,R,V)$ where $W$ is a non empty set of \emph{possible worlds}, $R$ a binary relation of $W$ and $V$ a function assigning a subset of $W$ with $p\in\At$. The pair $(W,R)$ is called a \emph{frame}. A frame, or a model, is said to be $n$-dense iff $R\subseteq R^n$.\\
A set of formulas $s$ is true at a world $w\in W$, denoted by $M,w\models s$ iff for some $u\in\CCS(s)$, all formulas of $u$ are classically true at $w$ and for all $\Box\phi\in u$ and for all $(w,w')\in R\colon M,w'\models \phi$~\footnote{And, as a consequence, for all $\neg\Box\phi\in u$ there exists $(w,w')\in R\colon M,w'\models \neg\phi$}. Concerning classical connectives, the usual induction-based definition of $\models$ is fulfilled by the notion of $\CCS$. Finally, $s$ is satisfiable in the class of $n$-dense models (or $n$-dense satisfiable) iff there exists an $n$-dense model $M=(W,R,V)$ and $w\in W$ such that $M,w\models s$. By the famous Salqvist's theorem, $n$-dense satisfiability and $\Log{n}$-consistency are known to be equivalent. Hence, we will talk of $\Log{n}$-satisfiability. A formula $\phi$ is valid iff $\neg\phi$ is not satisfiable. 

We establish some useful facts about $\CCS$: 
\begin{proposition}\label{prop-CCS}
For finite $u, v, s, s_1, s_2\subseteq \Fo$,
\begin{enumerate}
\item\label{One} if $u\in \CCS(s\cup v)$ and $v\in \CCS(s')$ then $u\in \CCS(s\cup s')$,
\item\label{Two} if $u\in \CCS(s\cup s')$ then it exists $v\in \CCS(s)$ and $v'\in\CCS(s')$ s.th.\ \ $v\cup v'=u$,
\item\label{Three} if $u\in \CCS(s'\cup v)$ and $v$ is a \CCS\ then it  exists $v'\in \CCS(s')$ s.th.\ $v\cup v' = u$,
\item\label{Four} if $u\in \CCS(s'\cup v)$ and $v$ is a \CCS\ then $d(u\setminus v)\leq d(s')$,
\item\label{Five} if $u$ is true at a world $x\in s$ of a $\Log{n}$-model $M$, then 
the set $\SF(u)\cap \{\phi\colon M,x\models \phi\}$ is in $\CCS(u)$.
\end{enumerate}
\end{proposition}
\begin{proof}
Item~(\ref{One}) is an immediate consequence of the properties of classical open branches of tableaux.
As for Item~(\ref{Two}), take $v=s\cap\CSF(u)$ and $v'=s'\cap\CSF(v)$.
Item~(\ref{Three}) follows from Item~(\ref{Two}).
Concerning Item~(\ref{Four}), if $u\in \CCS(s'\cup v)$ then by Item~(\ref{Three}), there exists $v'\in \CCS(s')$ and $v\cup v' = u$.
Therefore, $u\setminus v\subseteq v'$ and $d(u\setminus v)\leq d(v')=d(s')$.
Finally, about Item~(\ref{Five}), the reader may easily verify it 
by applying the definition clauses of $\models$. 
\end{proof}
An interesting point about $n$-dense models and frames, is that they are under disjoint union: if $(W,R)$ and $(W',R')$ are $n$-dense frames, then $(W,\sqcup W',R\sqcup R')$ is an $n$-dense frame too\footnote{Since $R\sqcup R'\subseteq R^n\sqcup R'n= (R\sqcup R')^n$.}. This is important to break the search for a model into independent subroutines. 

\section{\emph{n}-dense logics are PSPACE-hard}\label{pspace-hardness}
For all atoms $p$, let $\tau_{p}:\ \Fo\longrightarrow\Fo$ be the function inductively defined as follows:
\begin{itemize}
\item $\tau_{p}(q)=q$,
\item $\tau_{p}(\bot)=\bot$,
\item $\tau_{p}(\neg\phi)=\neg\tau_{p}(\phi)$,
\item $\tau_{p}(\phi\wedge\psi)=\tau_{p}(\phi)\wedge\tau_{p}(\psi)$,
\item $\tau_{p}(\square\phi)=\square(p\rightarrow\tau_{p}(\phi))$.
\end{itemize}
Obviously, for all atoms $p$ and for all $\phi\in\Fo$, $\pipe\tau_{p}(\phi)\pipe\leq5.\pipe\phi\pipe$.
\begin{lemma}\label{lemma:about:the:translation}
For all atoms $p$ and for all formulas $\phi$, if $p$ does not occur in $\phi$ then the following conditions are equivalent:
\begin{enumerate}
\item $\phi$ is valid in the class of all frames,
\item $\tau_{p}(\phi)$ is valid in the class of all frames,
\item $\tau_{p}(\phi)$ is valid in the class of all dense frames.
\end{enumerate}
\end{lemma}
\begin{proof}
Let $p$ be an atom and $\phi$ be a formula.
Suppose $p$ does not occur in $\phi$.
Obviously, $\mathbf{(2)\Rightarrow(3)}$.
Consequently, it suffices to prove that $\mathbf{(1)\Rightarrow(2)}$ and $\mathbf{(3)\Rightarrow(1)}$.\\
$\mathbf{(1)\Rightarrow(2):}$
Suppose $\tau_{p}(\phi)$ is not valid in the class of all frames.
Hence, there exists a model $M=(W,R,V)$ and there exists $s\in W$ s.th.\ $M,s\not\models\tau_{p}(\phi)$.
Let $M^{\prime}=(W^{\prime},R^{\prime},V^{\prime})$ be the model s.th.\
\begin{itemize}
\item $W^{\prime}=W$,
\item for all $t,u\in W$, $tR^{\prime}u$ if and only if $tRu$ and $u\in V(p)$,
\item for all atoms $q$, $V^{\prime}(q)=V(q)$.
\end{itemize}
As the reader may easily verify by induction on $\psi\in\Fo$, if $p$ does not occur in $\psi$ then for all $t\in W$, $M,t\models\tau_{p}(\psi)$ if and only if $M^{\prime},t\models\psi$.
Since $p$ does not occur in $\phi$ and $M,s\not\models\tau_{p}(\phi)$, then $M^{\prime},s\not\models\phi$.
Thus, $\phi$ is not valid in the class of all frames.\\
$\mathbf{(3)\Rightarrow(1):}$
Suppose $\phi$ is not valid in the class of all frames.
Consequently, there exists a model $M=(W,R,V)$ and there exists $s\in W$ such that $M,s\not\models\phi$.
Without loss of generality, suppose $W$ and $R$ are disjoint.
Let $M^{\prime}=(W^{\prime},R^{\prime},V^{\prime})$ be the dense model such that
\begin{itemize}
\item $W^{\prime}=W\cup R$,
\item for all $t,u\in W$, $tR^{\prime}u$ iff $tRu$,
\item for all $t\in W$ and for all $(u,v)\in R$, $tR^{\prime}(u,v)$ iff $t=u$
\item for all $(t,u)\in R$ and for all $v\in W$, $(t,u)R^{\prime}v$ iff $u=v$,
\item for all $(t,u),(v,w)\in R$, $(t,u)R^{\prime}(v,w)$ iff $t=v$ and $u=w$,
\item for all atoms $q$, if $q\not=p$ then $V^{\prime}(q)=V(q)$.
\item $V^{\prime}(p)=W$.
\end{itemize}

As the reader may easily verify by induction on $\psi\in\Fo$, if $p$ does not occur in $\psi$ then for all $t\in W$, $M,t\models\psi$ if and only if $M^{\prime},t\models\tau_{p}(\psi)$.
Since $p$ does not occur in $\phi$ and $M,s\not\models\phi$, then $M^{\prime},s\not\models\tau_{p}(\phi)$.
Thus, $\tau_{p}(\phi)$ is not valid in the class of all dense frames.
\end{proof}
\begin{proposition}
$n$-dense satisfiability is $\PSPACE$-hard.
\end{proposition}
\begin{proof}
By Lemma~\ref{lemma:about:the:translation} and the fact that the validity problem in the class of all frames is $\PSPACE$-hard~\cite[Theorem~$6.50$]{Blackburn:deRijke:Venema}.
\end{proof}

\section{Tableau calculus}\label{tableaux}

As said above, tableaux are a powerful way of establishing complexity results for a wide range of logics by means of decision procedures consisting in (un)successful attempts to proving the existence of a model. In the sequel, a tableau for a set $s$ will be defined as a relational structure $(W,R,F)$ where $W$ is a non empty set of \emph{nodes}, $R$ a binary relation on $W$ rooted in $root$, and $F$ a function assigning a set of formulas with each element of $W$, this structure having additional features according to the underlying logic. It is quite straightforward to design a tableau calculus for $\Log{n}$: it suffices to mimic the semantics: use $\CCS$ for the classical part, add successor node to $w$ for each $\neg\Box\phi\in w$, ensures $\phi\in w'$ for each $(w,w')\in R$ with $\Box\phi\in F(w)$ (this is usually referred to as \emph{$\Box$-rule}), and add $n-1$ intermediary nodes $wRw_1\cdots Rw_{n-1}Rw'$ between $w$ and $w'$ whenever $(w,w')\in R$, but the result is generally infinite (infinitely many intermediary nodes need to be created). A tableau is saturated iff all possible semantical constraints are respected. 

We will not make use of the usual notion of closed tableaux (containing $\bot$), in our setting we will rather refer to existence of tableaux.

In the sequel, when no confusion may arise, we will write $x$ instead of $F(x)$, e.g.\ in $\forall \phi\in x$ instead of $\forall \phi\in F(x)$. 
\subsection{Completeness and soundness}

In proofs below, we will use an extension of $\Log{n}$-satisfiability over tableaux as follows: given a tableau $T=(W,R,F)$ we will say that $T$ is $\Log{n}$-satisfiable iff there exists a model $M'=(W',R',V')$ such that $W\subseteq W'$, $R\subseteq R'$ and for all $w\in W$, for all $\phi\in F(w)\colon M,w\models \phi$ (note that in case $W$ is a singleton, it just amounts to $\Log{n}$-satisfiability).

Given a set of formulas $s$, we consider a function $\chooseCCS$ that  non deterministically picks a set in $\CCS(s)$ if possible (otherwise $s$ is not classically consistent, and hence is unsatisfiable)\footnote{Such a function can directly be obtained from a SAT-prover.}. We begin by giving a naive (non terminating) tableau calculus for $\Log{n}$:

\setlength{\textfloatsep}{0pt}
\setlength{\floatsep}{0pt}
\begin{algorithm}
 \floatname{algorithm}{Function}
\begin{algorithmic}
\caption{Tableau for the set $s$: pick some $u$ from $\CCS(s)$ and call $\sat(W=\{root\},R=\emptyset,F=\{root\mapsto u\})$}
\Function{\sat}{$W,R,F$}:
\State $W_0=\emptyset$
\While {$W\neq W_0$ (i.e.\ $W$ has been modified)}
\State $W_0:=W$
\ForAll {$w\in W, \neg\Box\phi\in w$}
\State {add new node $w^{\neg\phi}$ to $W$}
\State {pick $u\in\CCS(\Box^-(F(w))\cup\{\neg\phi\})$}
\State {add $(w,w^{\neg\phi})$ to $R$ and set $F(w^{\neg\phi})=u$}
\EndFor
\ForAll {$(w_0,w_n)\in R$ with $(w_0,w_n)\not\in R^n$}
\State let $u_0=F(w_0)$ and $u_n=F(w_n)$
\State {add new nodes $w_1\cdots w_{n-1}$ to $W$}
\State { pick $(u_i)_{1\leq i\leq n-1}\in\CCS^{n-1}$ s.th. $\Box^-(u_{i-1})\subseteq u_i$ \\ \phantom{mm} and $u_{n-1} \subseteq u_n$ and set $F(w_i)=u_i$}
\State {add $(w_0,w_1)\cdots(w_{n-1},w_n)$ to $R$}
\EndFor
\EndWhile
\EndFunction
\end{algorithmic}
\end{algorithm}
NB: if a call to ```pick'' fails, then the function returns ``unsatisfiable''. In fact, this function is just a reformulation of standard ones, like that of \cite{Baldoni2}. Classically, the function consists in infinite application of 1) developing  all $\Diamond$-formulas and 2) creating all necessary intermediary nodes, all filled by the $\Box$-rule. This provides a non terminating algorithm which is sound and complete:  there exists a saturated tableau for $s$ iff $s$ is satisfiable.
\begin{lemma}
The above algorithm is sound and complete: it provides a saturated tableau for $u$ if and only if $u$ is satisfiable. 
\end{lemma}
\begin{proof}[Sketch, see e.g.\ \cite{Baldoni2} for more details]
    Soundness is trivial: a saturated tableau directly provides a model $M=(W,R,V)$ with $w\in V(p)$ iff $p\in F(w)$. Completeness is easy: the truth lemma is straightforward: $\forall\phi\in\SF(s)\forall u\in W\colon \phi\in F(u)$ iff $M,V\models \phi$. The only non-straightforward point is in proving that given $(w_0,w_n)\in R$, there are indeed the desired $(u_i)_{1\leq i\leq n-1}$. The argument is classical: consider the induction hypothesis that \emph{at step $i$ the tableau is $\Log{n}$-satisfiable} in some model $M$: if there were no such $u_i$s this would imply that $\forall u_i\in\CCS$ (
    $1\leq i<n-1$) $\forall  u_{n-1}\in\CCS\colon \Box^-(u_{n-1})\vdash \neg u_n$ and thus $u_{n-1}\vdash \Box\neg u_n$; then, by axiom $\K$ and Necessitation since $u_{n-1}\in\CCS(u_{n-2})$, it comes: $\Box^-(u_{n-2})\vdash \Box\neg u_n$, thus $u_{n-2}\vdash \Box^2\neg u_n$, and so on until $u_0\vdash\Box^n\neg u_n$. Then by axiom of $n$-density we have $u_0\vdash\Box\neg u_n$, a contradiction since in $M$ the formulas $\Diamond\neg u_n$ should also be true at $w_0$. Hence, the desired $u_i$s exist. 
\end{proof}

\section{Windows, informally}\label{simple_windows}
Until section \ref{general_case}, let us focus on  the $2$-dense case. We begin by giving intuition on the arguments. 

By construction, this structure is a dense directed acyclic graph, rooted in $root$, but the above algorithm shows also, on the fly, that $\Log2$ tableaux are partly made of the following shapes (Fig.\ref{part_model}):
\begin{figure}[!h]
\begin{centering}
  \fbox{
  \begin{tikzpicture}[scale=0.9]
\node (u) at (2,1) {$u$};
\node (u0) at (2,0) {$u_0$};
\node (v0) at (1,-0.5)  {$v_0$};
\node (u1) at (0.75,0.25) {$u_1$};
\node (u2) at (-0.5,0.5) {$u_2$};

\draw [->,>=latex] (u2) -- (u1);
\draw [->,>=latex] (u1) -- (u0);
\foreach \n in {u0,u1,u2} \draw [->] (u)  to[bend right=10](\n);
\draw [->,>=latex](u1) -- (v0) ;
\draw [->,>=latex](v0) -- (u0) ;
\end{tikzpicture}
}
\caption{Part of the intermediary nodes between $u$ and $u_0$\\}\label{part_model}
\Description{Node $u_0$ links to $u_2,u_1,u_0$, each $u_{i+1}$ links to $u_i$, $u_1$ to $v_0$, and $v_0$ to $u_0$}
\end{centering}
\end{figure}
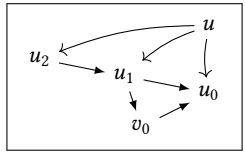

But there is more. Let us define the degree of an edge $(u,v)$ by $d(u,v)=d(u)$ and define $R(n)=\{(u,v)\in R\colon d(u,v)\leq n\}$, then a tableau verifies the following refinement of $2$-density $R(n)\subseteq R(n)\circ R(n-1)$ which will be the basis of our algorithm. With the nodes of Fig.\ \ref{part_model} it gives: fulfilling conditions for an $n$-edge $(u,u_0)$ consists in 1) satisfying $u_0$, 2) fulfilling conditions for an(other) $n$-edge $(u,u_1)$ and 3) for an $n-1$-edge $(u_1,v_0)$. Maybe the reader can already feel that the last part is a recursive call (that will end with $n=0$) as well as the second (that will end with empty nodes), and the first part will lead to a loop-detection argument. \\
But if we aim at breaking these three parts into independent ones, we must overcome one problem: choices made for the 2) may interfere with those for 3) which in their turn may interfere with those for 1), all because of the $\Box$-rule. The way to overcome this, is to consider long enough sequences of nodes. As shown in Fig.\ \ref{part_model1}: 

\begin{figure}[!h]
\begin{centering}
  \fbox{
  \begin{tikzpicture}[scale=0.9]
\node (u) at (2,1) {$u$};
\node (u0) at (2,0) {$u_0$};
\node (v0) at (1,-0.5)  {$v_0$};
\node (u1) at (0.75,0.25) {$u_1$};
\node (u2) at (-0.3,0.45) {$u_2$};
\node (v1) at (-0.05,-0.25)  {$v_1$};
\node (un) at (-1.4,0.68) {$u_n$};

\draw [dotted,->,>=latex] (un) -- (u1);
\draw [->,>=latex] (u1) -- (u0);
\foreach \n in {u0,u1,un} \draw [->,>=latex] (u)  to[bend right=10](\n);
\draw [->,>=latex](u1) -- (v0) ;
\draw [->,>=latex](v0) -- (u0) ;
\draw [->,>=latex](u2) -- (u1) ;
\draw [->,>=latex](v1) -- (u1) ;
\draw [->,>=latex](u2) -- (v1) ;

\draw (2,0) -- (1.5,-1); 
\draw (2,0) -- (2.5,-1); 
\draw (2.5,-1) -- (1.5,-1); 
\node (z) at (2,-1.3) {Model of $u_0$};

\end{tikzpicture}
}
\caption{Influence of intermediary nodes between $u$ and $u_0$\\}\label{part_model1}
\Description{Node $u_0$ links to $u_n,u_1,u_0$, there is a $n-1$ path between $u_n$ and $u_1$, $u_{1}$ links to $u_0$, $u_1$ to $v_0$, and $v_0$ to $u_0$}
\end{centering}
\end{figure}
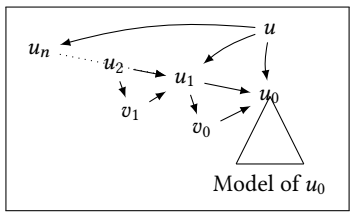
If $u_n$ is far enough from $u_1$ then it will not interfere, ``far enough'' means here that $n=d(u)$ in which case boxed formulas issued from $u_n$ can only ``go'' to $u_1$ and no further. Let us temporally call such a sequence $u_n,\cdots,u_0$ a $n$-window and consider the last problem which consists in the above mentioned ``loop-detection'': 
\begin{figure}[!h]
\begin{centering}
  \fbox{
  \begin{tikzpicture}[scale=0.7]
\node (u) at (2,1) {$u$};
\node (u0) at (2,-1) {$u_0$};
\node (ui) at (0.75,-0.75) {$u_i$};
\node (uj) at (-0.5,-0.5) {$u_j$};
\node (uip) at (-2.35,-0.2) {$u_{i'}$};
\node (ujp) at (-3.6,0.1) {$u_{j'}$};

\draw [dotted,->,>=latex] (uj) -- (ui);
\draw [dotted,->,>=latex] (ujp) -- (uip);
\draw [dotted,->,>=latex] (ui) -- (u0);
\draw [dashed,->,>=latex] (uip) -- (uj);
\foreach \n in {u0,ui,uj,uip,ujp} \draw [->,>=latex] (u)  to[bend right=10](\n);
\node[draw, dotted, rounded corners,rotate=-10, fit=(ui)(uj), inner sep=-1pt, label=below:{\small }] {};
\node[draw, dotted, rounded corners,rotate=-10, fit=(uip)(ujp), inner sep=-1pt, label=below:{\small }] {};
\end{tikzpicture}
}
\caption{Part of the nodes between $u$ and $u_0$. Boxes denote identical windows.\\}\label{part_model2}
\Description{A path with a loop of sequences}
\end{centering}
\end{figure}
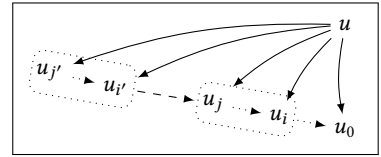
if we suppose $j-i=j'-i'\geq d(u)$ (i.e.\ windows are sufficiently long), then eventually, during the computation, there will be a repetition of windows (Fig.\ \ref{part_model2}), and precisely because they are long enough, we are sure that the dashed part can be repeated too (then the window, then the dashed part, and so on). This will be the backbone of the looping argument. Up to now, we just have taken a look back at windows as defined in \cite{BalGasq25,IGPL-Gasquet25} for simpler cases, but all those considerations must be applied also to subwindows and subsubwindows, etc. we need to generalize the notion of window. Now let us go into the details. 

\section{Recursive windows and properties}\label{windows}
All along this section, we will use $\lambda$ to denote either an integer function on a \CCS\ $u$, with $\lambda(u)\geq d(u)$ and $\lambda$ is increasing w.r.t.\ $d$ (i.e.\ $d(u)>d(v)\rightarrow \lambda(u)>\lambda(v)$), or $\lambda$ may be the constant function $\lambda(u)=\infty$ for all $u$. 

\subsection{Windows} 

Let $u$, $v_0$ be two \CCS, $k\leq d(u)$ and $n\geq d(u)$. 
A $(k,n,\lambda)$-window for $(u,v_0)$ denoted by $W$ is a pair $\pair{\mathcal N}{\mathcal W}$ with (${\mathcal N}$ for ``nodes'' and ${\mathcal W}$) for ``windows'':\\
$\bullet$ if $k=0$, $\pair{\mathcal N}{\cal W}=\voidpair$ (empty window)\\
$\bullet$ if $k>0$, $\pair{\mathcal N}{\cal W}$ is a pair of two sequences with: 
        \begin{itemize}
            \item [-]${\mathcal N}=\sequencef{v}{0}{\sizew}$ is a sequence of nodes s.th.:
        \begin{itemize}   
            \item [*]for each $i\in\intero{0}{\sizew}\colon v_i\in\CCS(\Bom(u)\cup\Bom(v_{i+1}))$
             \item [*]if $\sizew\neq \infty$ then $v_{\sizew}\in \CCS(\Bom(u))$
        \end{itemize}
        \item [-] ${\cal W}=\sequencef{W}{0}{\sizew}$ is a sequence of $(k-1,\lambda(v_{i+1}),\lambda)$-windows for $(v_i,u_i)$ and for $(u_{i+1},v_i)$ respectively, called subwindows of $W$. 
        \end{itemize}
Fig.\ \ref{part_model4} shows how it looks like at the top level. 
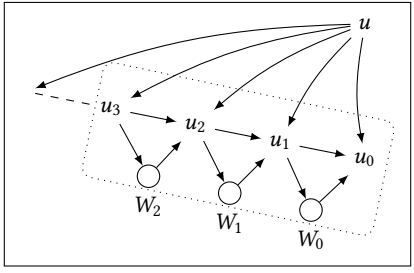
\begin{figure}[t]
\begin{centering}
  \fbox{
  \begin{tikzpicture}[scale=0.9]
\node (u) at (2,1) {$u$};
\node (u0) at (2,-1) {$u_0$};
\node (u1) at (0.75,-0.75) {$u_1$};
\node (u2) at (-0.5,-0.5) {$u_2$};
\node (u3) at (-1.75,-0.25) {$u_3$};
\node (u4) at (-3,0) {};

\node[draw,circle] (w0) at (1.2,-1.75) {};
\node[draw,circle] (w1) at (0,-1.5) {};
\node[draw,circle] (w2) at (-1.2,-1.25) {};

\draw [->,>=latex] (u2) -- (u1);
\draw [dashed,>=latex] (u4) -- (u3);
\draw [->,>=latex] (u1) -- (u0);
\draw [->,>=latex] (u3) -- (u2);
\draw [->,>=latex] (u1) -- (w0);
\draw [->,>=latex] (w0) -- (u0);
\draw [->,>=latex] (u2) -- (w1);
\draw [->,>=latex] (w1) -- (u1);
\draw [->,>=latex] (u3) -- (w2);
\draw [->,>=latex] (w2) -- (u2);
\draw (w0.south) node[below]{$W_0$};
\draw (w1.south) node[below]{$W_1$};
\draw (w2.south) node[below]{$W_2$};
\foreach \n in {u0,u1,u2,u3,u4} \draw [->,>=latex] (u)  to[bend right=10](\n);
\node[draw, dotted, rounded corners,rotate=-12, fit=(u0)(u1)(u2) (u3)(w1)(w2), inner sep=1pt, label=below:{\small }] {};
\end{tikzpicture}
}
\caption{A $(3,d(u),\lambda)$-window for $(u,u_0)$, hence of size 3. Subwindows $W_0$, $W_1$ and $W_2$ are circles. Size of subwindows is not visible, hence $\lambda$ is not specified \\}\label{part_model4}
\Description{A window composed of a sequence of nodes and a sequence of subwindows}
\end{centering}
\end{figure}

When the context is clear, we will sometimes omit the indices. Also, given a window $\langle\sequencef{v}{0}{\sizew},\sequencef{W}{0}{\sizew}\rangle$ we just need to say it is a window for $u$ (since $v_0$ may be seen from $\sequencef{v}{0}{\sizew}$). 

These sets and those of the subwindows will constitute a tableau, thus a window is a recursive structure which locally corresponds to a small part of an infinite tableau. 

The number $k$ represents the degree of  nesting of the window w.r.t.\ the main one: $k$ decreases as the nesting increases, but the reader may wonder why we also need $n$ and $\lambda$. Number $n$ can be thought of as the length of the main window, but in the course of expending a window, as we will see below, we will need to consider window whose length is different than that of its subwindows. Also, according to the case, we need subwindows of (linear) length (for the algorithm) and of (exponential) length (for completeness), thus the need of a function $\lambda$ which will essentially be either $d$ or some function $\chi$ whose definition is about to be given. 
We inductively define the set $\{W\}$ of members of a window by: 
\begin{itemize}
    \item $\{\voidpair\}=\emptyset$
    \item $\{\win{v}{W}{\sizew}\}=\{v_i\colon i\in[0:\sizew]\}\cup \bigcup_{i\in[0:\sizew[}\{W_i\}$
\end{itemize}

\noindent N.B.\ for any set $v\in \{W\}$, we have: $\lgx{v}\leq c_{\mathrm{sf}}.\lgu$ and $d(v)< d(u)$. 
\paragraph{Partial windows} 
Let $T=\win{v}{W}{\sizew}$ be a $(k,\sizew,\lambda)$-window for $(u,v_0)$ as above. Then 
$\partialwin{a}{v}{W}{b}$ (with $0\leq a<b\leq \sizew$) is a $(k,b-a+1,\lambda)$-partial window of $W$. Of course, $W$ is a partial window of $W$ (take $a=0$ and $b=\sizew$). 
\paragraph{Pointwise $k$-inclusion of partial windows} 
Let $u$, $v^1_0$ and $v^2_0$ be \CCS, let $0\leq k\leq d(u)$. Let $W^1$ be a $(k,\sizew,\lambda)$-window for $(u,v^1_0)$ and let $W^2$ be a $(k,\sizew,\lambda)$-window for $(u,v^2_0)$. Let $0\leq a<b\leq l$ and let $SW_1$ and $SW_2$ be two $(k,b-a+1,\lambda)$-partial windows of $W_1$ and $W_2$, we define the pointwise $k$-inclusion of the partial windows $SW_1\sqsubseteq^k SW_2$ by:\\
$\bullet \voidpair\sqsubseteq^{0} \voidpair$\\
$\bullet $ if $k>0\colon \partialwin{a}{v^1}{W^1}{b}\sqsubseteq^k \partialwin{a}{v^2}{W^2}{b}$ iff:
        \begin{itemize}
            \item for $i\in\intero{a}{b}\colon v^2_i\in\CCS(v^1_i\cup\Bom(v^2_{i+1}))$
            \item and for $i\in\intero{a}{b}\colon W^1_i\sqsubseteq^{k-1} W^2_i$ 
            \item [] N.B.\ if $d(u)=1$ then for $i\in\intero{a}{b}\colon d(\Bom(v^2_{i}))=0$ and hence $v^2_{i}=v^1_{i}$.
        \end{itemize}

\paragraph{Continuations of windows}
Let $u$, $v^1_0$ and $v^2_0$ be \CCS, let $0< k\leq d(u)$, and $d(u)\geq 1$. Let $W^1$ be the $(k,\sizew,\lambda)$-window for $(u,v^1_0)$ given by: $\win{v^1}{W}{\sizew}$  and $W^2$ be the $(k,\sizew,\lambda)$-window for $(u,v^2_0)$ given by: $ \win{v^2}{W}{\sizew}$. We say that $W_2$ is a $k$-continuation of $W_1$ iff 
$$\partialwin{1}{v^1}{W^1}{\sizew} \sqsubseteq^k \partialwin{0}{v^2}{W^2}{\sizew-1}$$
(beware of indexes: the ``end'' of $W_1$ is included ($\sqsubseteq^k$) in the ``beginning'' of $W_2$). 

We provide intuition of a continuation of a $(2,d(u),d)$-window in Fig.\ \ref{continuation} with $d(u)=1$ (surely $d(u)$ is ridiculously small, but this is for keeping the window small enough). Above is a window, and below one of its continuations (pointwise included by $\sqsubseteq$). As one can see, boxed formulas of nodes strictly between $\tilde v_2$ and $\tilde v_1$ cannot interfere with nodes between $v_0$ and $v_1$ (since they are of degree 0). Hence, provided the subwindow between the latter two is satisfiable, we can forget it and proceed to try to extend the window. For this, we have to test $\sqsubseteq$-inclusion of the $\bullet$-part into the $\circ$-part . The same reasoning applies at each scale on subwindows. \\

\begin{figure}[t]
\begin{centering}
  \fbox{
  \begin{tikzpicture}[scale=0.5]
\coordinate (u) at (4,8) ;
\coordinate (v0) at (16,8) ;
\coordinate[shift={(-20mm,-5mm)}] (v1) at (v0);
\coordinate[shift={(-20mm,-5mm)}] (v2) at (v1);
\coordinate (v0p) at (12,3) ;
\coordinate (v1p) at (8,2) ;
\coordinate (v2p) at (4,1) ;
\coordinate (v2b) at (4,-1) ;
\coordinate (vp112) at (8,0) ;
\coordinate (arr) at (10,3.5) ;

\coordinate (v02) at (12,6);
\coordinate (v01) at (14,7) ;
\coordinate (v12) at (8,5) ;
\coordinate (v11) at (10,6) ;
\coordinate (v001) at (15,7) ;
\coordinate (v002) at (14,6) ;
\coordinate (v012) at (12,5) ;
\coordinate (v011) at (13,6) ;
\coordinate (v101) at (11,6) ;
\coordinate (v102) at (10,5) ;
\coordinate (v111) at (9,5) ;
\coordinate (v112) at (8,4) ;

\coordinate (vp12) at (8,1) ;
\coordinate (vp11) at (10,2) ;
\coordinate (vp101) at (11,2) ;
\coordinate (vp102) at (10,1) ;
\coordinate (vp111) at (9,1) ;
\coordinate (vp112) at (8,0) ;

\draw (u) node{$\bullet$} node[above]{$u$};
\draw (v0) node{$\bullet$} node[above]{$v_0$};
\draw (v1) node{$\bullet$} node[above]{$v_1$};
\draw (v2) node{$\bullet$} node[above]{$v_2$};
\draw (v01) node{$\bullet$};
\draw (v02) node{$\bullet$};

\draw (v11) node{$\bullet$};
\draw (v12) node{$\bullet$};
\draw (v101) node{$\bullet$};
\draw (v102) node{$\bullet$};
\draw (v111) node{$\bullet$};
\draw (v112) node{$\bullet$};

\draw (vp11) node{$\circ$};
\draw (vp12) node{$\circ$};
\draw (vp101) node{$\circ$};
\draw (vp102) node{$\circ$};
\draw (vp111) node{$\circ$};
\draw (vp112) node{$\circ$};

\draw (arr) node{\rotatebox[origin=c]{-90}{$\sqsubseteq$ ?} };

\draw (v001) node{$\bullet$};
\draw (v002) node{$\bullet$};
\draw (v011) node{$\bullet$};
\draw (v012) node{$\bullet$};

\draw (v0p) node{$\circ$} node[above]{$\tilde v_0$};
\draw (v1p) node{$\circ$} node[above]{$\tilde v_1$};
\draw (v2p) node{$\circ$} node[above]{$\tilde v_2$};

\draw [dotted] (u) -- (v0);
\draw [dotted] (u) -- (v1);
\draw [dotted] (u) -- (v2);
\draw [dotted] (v2) -- (v1) -- (v0);
\draw [dotted] (v2) -- (v12);
\draw [dotted] (v1p) -- (vp12);
\draw [dotted] (v2) -- (v11);
\draw [dotted] (v1p) -- (vp11);
\draw [dotted] (v1p) -- (v0p);

\draw [dotted] (v12) -- (v112);
\draw [dotted] (vp12) -- (vp112);
\draw [dotted] (v12) -- (v111);
\draw [dotted] (vp12) -- (vp111);
\draw [dotted] (v12) -- (v11) -- (v1);
\draw [dotted] (vp12) -- (vp11) -- (v0p);
\draw [dotted] (v112) -- (v111) -- (v11);
\draw [dotted] (vp112) -- (vp111) -- (vp11);
\draw [dotted] (v11) -- (v102);
\draw [dotted] (vp11) -- (vp101);
\draw [dotted] (vp11) -- (vp102);
\draw [dotted] (v11) -- (v101);
\draw [dotted] (v102) -- (v101) -- (v1);
\draw [dotted] (vp102) -- (vp101) -- (v0p);
\draw [dotted] (v2) -- (v1) -- (v0);
\draw [dotted] (v1) -- (v02);
\draw [dotted] (v1) -- (v01);
\draw [dotted] (v02) -- (v01) -- (v0);
\draw [dotted] (v01) -- (v002);
\draw [dotted] (v01) -- (v001);
\draw [dotted] (v02) -- (v012);
\draw [dotted] (v02) -- (v011);
\draw [dotted] (v012) -- (v011) -- (v01);
\draw [dotted] (v002) -- (v001) -- (v0);

\draw [dotted] (v1p) -- (v2p) -- (v2b) -- cycle;
\end{tikzpicture}
}
\caption{A $(2,d(u),d)$-window and one of its potential continuation: arrows are left-to-right or else top-bottom (nodes between {$\tilde v_2$ and {$\tilde v_1$} are not represented).} }\label{continuation}
\Description{}
\end{centering}
\end{figure}
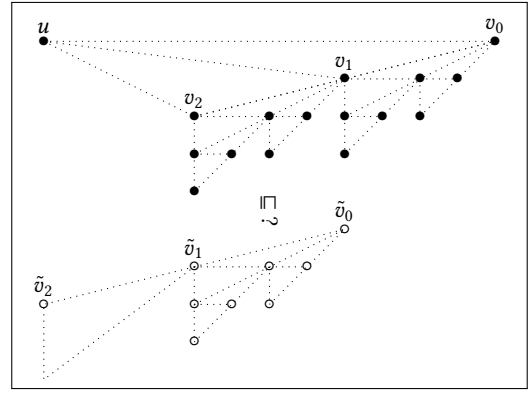
\begin{lemma}\label{k2kplusone-window}
    Let $u$, $v_0$ be two \CCS, let $0\leq k\leq d(u)$, and $d(u)\geq 1$. Let $W_1$ and $W_2$ be two $(k,\sizew,\lambda)$-windows for $u$ with:\\
    $W_1=\win{v^1}{W^1}{\sizew}$      and $W_2=\win{v^2}{W^2}{\sizew}$, and suppose $W_2$ is a $k$-continuation of $W_1$, then with
     \begin{itemize}
         \item $(z_0,D_0)=(v^1_0,W^1_0)$
         \item $z_{\sizew+1}=v^2_{\sizew}$
         \item and for $1<z\leq \sizew\colon (z_i,D_i)=(v^2_{i-1},W^2_{i-1})$
     \end{itemize}
          we have $\win{z}{D}{\sizew+1}$ is a $(k,\sizew+1,\lambda)$-window for $(u,v_0)$. 
\end{lemma} 
\begin{proof}
First we need to prove the following proposition (about the non-interference of nodes too far):
\begin{proposition}\label{proposition-degree}
For $1\leq i\leq \sizew\colon d(v^2_{i-1}\setminus v^1_i) \leq d(u)+i\dotminus (\sizew+1)$, 
by descending induction on $i\in\{1,\ldots,\sizew\}$: either $i=\sizew$, or $i<\sizew$.
In the former case, $v^2_{\sizew-1}\in \CCS(v^2_{\sizew}\cup v^1_{\sizew})$.
Since $v^2_{\sizew}\in \CCS(\Bom(u))$ and $v^1_{\sizew}\in \CCS(\Bom(u))$, then $v^1_{\sizew}\leq d(u)\dotminus 1$ and $d(\Bom(v^2_{\sizew}))\leq d(u)\dotminus 2$.
Consequently, $d(v_{i-1}^2\setminus v_{\sizew}^1) \leq d(u)\dotminus 1$. In the latter case, $v_{i-1}^2\in \CCS(\Bom(v_{i}^2)\cup v_{i}^1)$; and since $v_{i}^2=v_{i}^2\cup v_{i+1}^1=(v_{i}^2\setminus v_{i+1}^1)\cup v_{i+1}^1$, and $\Bom(A\cup B)=\Bom(A)\cup \Bom(B)$, we have $v_{i-1}^2\in \CCS(\Bom(v_{i}^2\setminus v_{i+1}^1)\cup \Bom(v_{i+1}^1)\cup v_{i}^1)$; but $\Bom(v_{i+1}^1)\subseteq v_{i}^1$, hence $v_{i-1}^2\in \CCS(\Bom(v_{i}^2\setminus v_{i+1}^1)\cup v_{i}^1)$. 
By Prop.\ \ref{prop-CCS}.\ref{Three}: $\exists v'\colon v'\in\CCS(\Bom(v_{i}^2\setminus v_{i+1}^1))$ and $v_{i-1}^2=v_{i}^1\cup v'$. Thus $v_{i-1}^2\setminus v_{i}^1\subseteq v'$, and  $d(v_{i-1}^2\setminus v_{i}^1)\leq d(v')=d(\Bom(v_{i}^2\setminus v_{i+1}^1))\leq d(v_{i}^2\setminus v_{i+1}^1)\dotminus 1 \leq d(u)+(i+1)\dotminus (\sizew+1)\dotminus 1$ (by induction hypothesis) $\leq d(u)+i\dotminus (\sizew+1)$.  \end{proposition}

Now, we check that $\longwin{z}{D}{l+1}{l}$ is indeed a $(k,\sizew+1,\lambda)$-window for $(u,v_0)$ by examining the definition of continuations. \\
Since we have $z_0,z_1,z_2,\cdots,z_{\sizew+1}=v^1_0,v^2_0,v^2_1,\cdots,v^2_{\sizew}$, and also have $D_0,D_1,D_2,\cdots,D_{\sizew}=W^1_0,W^2_0,W^2_1,\cdots,W^2_{\sizew-1}$, it comes:
\begin{enumerate}
    \item $v_{\sizew}^2 \in \CCS(\Bom(u))$
    \item 
    \begin{enumerate}
    \item $v_{\sizew-1}^2\in\CCS(v_{\sizew}^1\cup\Bom(v_{\sizew}^2))$, and since $v_{\sizew}^1\in\CCS(\Bom(u))$, it comes $v_{\sizew-1}^2\in\CCS(\Bom(u)\cup\Bom(v_{\sizew}^2))$;
    \item take $i\in[0:\sizew-2]$. Then $v_{i-1}^2\in\CCS(v_i^1\cup\Bom(v_i^2))$, and since $v_i^1\in\CCS(\Bom(u)\cup\Bom(v_{i+1}^1)$, by Prop.\ \ref{prop-CCS}.\ref{One}: $v_{i-1}^2\in\CCS(\Bom(u)\cup\Bom(v_i^2)\cup\Bom(v_{i+1}^1))$. But $v_{i+1}^1\inc v_i^2$, hence $v_{i-1}^2\in\CCS(\Bom(u)\cup\Bom(v_i^2))$; 
    \item to conclude for condition 2, it remains to prove that $v_0^1\in \CCS(\Bom(v_0^2)\cup \Bom(u))$. By Prop.\ \ref{proposition-degree}, $d(v_0^2\setminus v_1^1)\leq d(u)\dotminus\sizew\leq 0$, hence if $\Box \phi\in v_0^2$ then $\Box \phi\in v_1^1$ and thus $\Bom(v_0^2)=\Bom(v_1^1)$.\\ Since $v_0^1\in \CCS(\Bom(v_1^1)\cup \Bom(u))$, then $v_0^1\in \CCS(\Bom(v_0^2)\cup \Bom(u))$.
    \end{enumerate}
    \item We verify condition 3 by proving that for each $i\in\intero{0}{\sizew}$ $W^2_i$ is a $(k+1,\lambda(v^2_{i+1}),\lambda)$-window for $(v^2_{i+1},v^2_i)$ and that $W^1_0$ is a $(k+1,\lambda(v^2_0),\lambda)$-window for $(v^2_0,v^1_0)$. It is immediate for the first ones. 
    Concerning $W^1_0$: if $k=1$ and $W^1_0=\voidpair$ then we are done, else let $W^1_0=\langle(v^{0,1}_i),(W^{0,1}_i)\rangle$ (we omit the ranges of the sequences) ; since it is a $(k+1,\lambda(v^1_1),\lambda)$-window for $(v^1_1,v^1_0)$ we have, for all $0\leq i< \lambda(v^1_1)\colon v^{0,1}_i\in\CCS(\Bom(v^{0,1}_{i+1})\cup\Bom(v^1_1))$, but recall from 3.c above that $\Bom(v^2_0)\inc\Bom(v^1_1)$, and since each $W^{0,1}_i$ is a $(k+2,\lambda(v^{0,1}_{i+1}),\lambda)$-window for $(v^{0,1}_{i+1},v^{0,1}_i)$, hence conditions are met to state that $W^1_0$ is a $(k+1,\lambda(v^2_0),\lambda)$-window for $(v^2_0,v^1_0)$. 
    \end{enumerate}
    
\end{proof}
\begin{lemma}\label{prolongation-to-infinite-window}
    Let $u$, $v_0$ be two \CCS, let $0\leq k\leq d(u)$, if there exists a $(k,\chi(u),d)$-window for $(u,v_0)$ for a ``sufficiently large'' $\chi(u)$ which depends on $u$, then there exists a $(k,\infty,d)$-window for $(u,v_0)$. Such a window will be called ``maximal''. 
\end{lemma}
\begin{proof}
    In order to precise ``sufficiently large", let us first compute the number of $\CCS$ in a $(k,d(u),d)$-window which is either $W=\win{v}{W}{d(u)}$ or $W=\voidpair$ according to $k$. This number is bounded above by the following recurrent inequalities: 
    $$\begin{array}{ll}
         s(k) &  = (\mbox{if } n=0 \mbox{ then } 1 \mbox{ else }{d(u)}+\Sigma_{i=0}^{d(u)-1}s(k-1)\\
         & \leq (\mbox{if } k=0 \mbox{ then } 1 \mbox{ else } d(u) + d(u).s(k-1)\\
         &\leq d(u) + d^2(u)+ d^2.s(k-1)\\
         &\leq d(u) + d^2(u)+ \cdots +d^k(u)+ d^k.s(0)\\
         &\leq Q(d(u)) \hfill\mbox{\indent for some polynomial $Q$ of degree $k$}\\
         &\leq Q(\lgu)\\
    \end{array}$$

    Each \CCS\ is a member of $\SF(u)$ and there are at most $2^{\lgx{\SF(u)}}=2^{c_{\mathrm{sf}}.\lgu}$ of them. Hence, there are at most $(2^{c_{\mathrm{sf}}.\lgu})^{Q(\lgu)}$ distinct $(k,d(u),d))$-windows for $(u,v_0)$, i.e.\ $2^{P(\lgu)}$ for some polynomial $P$ of degree $k+1$. \\
We claim that $\chi(u)=2^{P(\lgu)}+d(u)$: let $W$ be a $(k,\uplim{u},d(u))$-window for $(u,v_0)$, it can be broken into the sequence $(W_j)_{j\in[0\:\uplim{u}]}$ of $(k,d(u),d))$-partial windows for $(u,v_0)$ each of them being a $k$-continuation of the previous. Then, because of the above bound, at least two of them are identical: there exists integers $h,\delta$ such that $\delta\neq 0$ and $h+\delta \leq \uplim{u}$ and $W_h=W_{h+\delta}$. Let $(\Wpp_j)_{0\leq j}$ be the infinite sequence such that for all $j\leq h$, $\Wpp_j=W_j$ and for all $j>h$, $\Wpp_j=W_{h+((j-h)\!\!\!\mod \delta)}$. By construction, for all $j\geq 0$, $\Wpp_{j+1}$ is a $k$-continuation of $\Wpp_j$. For all $j\geq 0$, suppose that $\Wpp_j=\win{v^j}{W^j}{d(u)}$, and set $W=\langle (v^i_0)_{0\leq i}, (W^i_0)_{0\leq i}\rangle$ which is a $(k,\infty,d)$-window for $(u,v_0)$. 
    \end{proof}

\begin{lemma}\label{from-model-to-window}
    Let $M=(S,R,V)$ be a $2$-dense model. Let $u,v_0$ be two \CCS\ and suppose that there exists $x,y_0\in S$ such that: $(x,y_0)\in R$ and $M,x\models u$ and $M,y_0\models v_0$. Then for any integer $k$, there exists $W$ a $(k,\chi(u),\chi)$-window and all its $\CCS$ are $\Log2$-satisfiable (i.e.\ for $(u,v_0)$ such that for all $v\in W\colon v$ is $\Log2$-satisfiable). 
\end{lemma}

\begin{proof}
By induction on $k$:
\begin{itemize}
    \item If $k=0$, we are done with $\voidpair$. 
    \item If $k>0$, since $(x,y_0)\in R$ and $M$ is $2$-dense, let $(y_i)_{i\geq 0}$ be the $(k,\infty)$-sequence for $(x,y_0)$ in $(S,R)$.\\
    For each $i\geq 1$ let $v_i=\SF(\Bom(u))\cap y_i$. Trivially $v_i$ is $\Log2$-satisfiable. \\
    First, we establish the following fact: for all $i\geq 0$, $v_i\in\CCS(\Bom(u)\cup\Bom(v_{i+1}))$: 
    \begin{enumerate}
        \item since $M,x\models u$ and $(x,y_i)\in\Ri{k}$, then $M,y_i\models \Bom(u)$, hence $\Bom(u)\inc \SF(\Bom(u))\cap y_i$, thus $\Bom(u)\inc v_i$, 2) let $\phi\in\Bom(v_{i+1})$, then $\Box\phi\in v_{i+1}$, hence $\Box\phi\in\SF(\Bom(u))\cap y_{i+1}$, thus $\phi\in\SF(\Bom(u))\cap y_{i}$, i.e.\ $\phi\in v_i$; finally, $\Bom(u)\cup\Bom(v_{i+1})\inc v_i$; 
        \item $v_i$ is saturated; we only consider the $\wedge$ case: let $(\phi\wedge\psi)\in s_i$, hence $(\phi\wedge\psi)\in\SF(\Bom(u))\cap y_i$, hence $\phi\in\SF(\Bom(u))\cap y_i$ (as well for $\psi$), hence $\phi\in y_i$ (idem for $\psi$); 
        \item being a finite subset of the consistent set $y_i$, $v_i$ is consistent. 
    \end{enumerate}
    1.\ to 3.\ together prove the fact. \\
    Thus, for each $i\geq 0$, both $v_{i+1}$ and $v_i$ are \CCS\ such that there exists $y_{i+1}$ and $y_i$ with $(y_{i+1},y_i)\in\Ri{k+1}$ and $M,y_{i+1}\models v_{i+1}$ and $M,y_i\models v_i$, hence induction hypothesis applies: there exists a $(k-1,\uplim{v_{i+1}},\chi)$-window for $(v_{i+1},v_i)$ with all its $\CCS$ $\Log2$-satisfiable. Let us denote it $W_i$. Finally, and since by hypothesis $v_0$ is $\Log2$-satisfiable, $\win{v}{W}{\sizew}$ is the desired $(k,\uplim{u},\chi)$-window for $(u,v_0)$. 
\end{itemize}
\end{proof}

\begin{corollary}\label{existence-of-window}
Let $u$ a \CCS\ containing some formula $\neg \Box\phi$ and satisfied at a world $x$ of a $2$-dense model $M$ then a) there exists $v_0\in\CCS(\{\neg\phi\}\cup\Bom(u))$ and b) there exists a $(d(u),\uplim{u},\chi)$-window for $(u,v_0)$ and all its $\CCS$ are $\Log2$-satisfiable.
\end{corollary}
\begin{proof}
    Since $M,x\models u$ then $M,x\models \neg\Box\phi\wedge \Box(u)$, hence there exists $y_0$: $(x,y_0)\in R$ and $M,y_0\models \neg\phi\wedge \Bom(u)$. Let $v_0=\SF(u)\cap y_0$ and conclude with the above lemma with $k=d(u)$. 
\end{proof}
\section{The algorithm}\label{algorithm}

The idea is that despite the infinity of $2$-dense models, and because of lemma \ref{prolongation-to-infinite-window}, it would suffices to check $(k,\uplim{u},\chi)$-windows. But they are of unparameterized exponential size, so we need to check them by exploring relatively small pieces at a time and this will appear to be recursively possible, thanks to continuations.

The algorithm we present below is based on the function $\sat$ which answers to the $\Log2$-satisfiability of its argument. 
Because of Prop. \ref{prop-CCS}.\ref{Five}, the $\Log2$-satisfiability of a set $s$ of formulas amounts to that of at least one of its $\CCS$, since $s$ is $\Log2$-satisfiable if and only if there exists a $\Log2$-satisfiable $u\in\CCS(s)$. Hence, given an initial set of formulas $s$ to be tested, the initial call is $\sat(\chooseCCS(\{s\}))$.\\
In what follows we use built-in functions \algand\ and \all.
The former function lazily implements a logical ``and".
The latter function lazily tests if all members of its list argument are true. 
\setlength{\textfloatsep}{0pt}
\setlength{\floatsep}{0pt}
\begin{algorithm}
 \floatname{algorithm}{Function}
\begin{algorithmic}
\caption{Test for $\Log2$-satisfiability of a \CCS: it must be classically consistent and recursively so for each $\Diamond$-formula subsequent window.}
\Function{\sat}{$u$}:
\State {return}
\State {\hspace{0.3cm}$u\neq \{\bot\}$}
\State {\hspace{0.1cm}$\algand$}
\State {\hspace{0.2cm}$\all \{\satW(\chooseW(u,\chooseCCS(\{\neg\phi\}\cup\Bom u,d(u)), u,d(u),\uplim{u})$}
\State {\hspace{1cm}$\colon\neg\Box\phi\in u\}$}
\EndFunction
\end{algorithmic}
\end{algorithm}

\begin{algorithm}
\floatname{algorithm}{Function}
\begin{algorithmic}
\caption{Returns $\{\bot\}$ if $s$ is not classically consistent, otherwise returns a $\CCS$, non-deterministically chosen}
\Function{\chooseCCS}{$s$}
\If  {$\CCS(s)\neq \emptyset$}
\State {return one $u\in \CCS(s)$}
\Else 
\State {return $\{\bot\}$}
\EndIf
\EndFunction
\end{algorithmic}
\end{algorithm}
\hspace{0.3cm}

\begin{algorithm}
\floatname{algorithm}{Function}
\begin{algorithmic}
\caption{Non-deterministically picks a $(k,d(u),d)$-window for $(u,v)$}
\Function{\chooseW}{$u$,$v$,$k$}
\If {there exists a $(k,d(u),d)$-window $W$ for $(u,v)$}
\State {return $W$} \algorithmiccomment{with $W=\win{v}{W}{d(u)}$}
\State {} \algorithmiccomment{or $\voidpair$ if $k=0$}
\Else 
\State {return $\langle (\{\bot\})_{i\in[0:d(u)]}, \emptyset\rangle $}
\EndIf
\EndFunction
\end{algorithmic}
\end{algorithm}
\hspace{0.3cm}

\begin{algorithm}
\floatname{algorithm}{Function}
\begin{algorithmic}
\caption{Tests the satisfiability of a $(k,d(u),d)$-window for $(u,v_0)$ and recursively that of each of its subwindows and continuations, until a repetition happens or a contradiction is detected}
\Function{\satW}{$W$,$u$,$k$,$N$}:\algorithmiccomment{$W$ is $\langle(v_i),(W_i)\rangle$ if $k>0$}
\State {} \algorithmiccomment{or $\voidpair$ if $k=0$}
\State {}   \algorithmiccomment{or $\langle (\{\bot\}), \emptyset\rangle$ if there were no possible $W$}
\If {$N=0$ \textbf{or} $d(u)=1$} 
\State {return \true} 
\Else 
\State {return}
\State {\hspace{0.83cm} \sat$(v_0)$}
\State {\hspace{0.2cm} $\algand\ \satW(W_0,v_{1},k-1,\uplim{v_1})$}\
\State {\hspace{0.2cm} \algand\ \satW(\nextW$(W,u,k),u,k,N-1))$}
\EndIf
\EndFunction
\end{algorithmic}
\end{algorithm}
\hspace{0.3cm}
\begin{algorithm}
\floatname{algorithm}{Function}
\begin{algorithmic}
\caption{Non-deterministically chooses a $k$-continuation of a window for $s$}
\Function{\nextW}{$W$,$u$}
\If {there exists a $k$-continuation $W_1$ of $W$ for $u$}
\State {return $W_1$}
\Else 
\State {return $\langle (\{\bot\})_{0\leq i\leq d(u)}, \emptyset\rangle $}
\EndIf
\EndFunction
\end{algorithmic}
\end{algorithm}

\section{Analysis of the algorithm}\label{analysis}

\begin{proposition}\label{proposition-on-members}
\mbox{}\\
\begin{enumerate}
    \item For all $j\in\interf{0}{\uplim{u}}$ we have $\{W\lbr  j\rbr \}\subseteq \{W\}$. 
    \item Given an initial call $\satW(W'\lbr  0\rbr ,u',k',\chi(u'))$, then in all subsequent calls $\satW(W\lbr  j\rbr ,u,k,N)$ the precondition $\{W\}\subseteq \{W'\}$ is satisfied. 
\end{enumerate}
\end{proposition}
\begin{proof} 
By induction on the recursion depth:
\begin{enumerate}
    \item
    If $W=\voidpair$, then we are done since $\{\voidpair\lbr j\rbr \}=\emptyset\subseteq\{W\}$\\
    else, suppose $W\lbr j\rbr =\shrinkwin{j}{v}{W}{j+d(u)}$; then $\{W\lbr j\rbr \}=\{v_i\colon i\in\interf{j}{j+d(u)}\cup\bigcup_{i\in[j:j+d(u)[}\{W_i\lbr 0\rbr \}$. Since a) $\{v_i\colon i\in\interf{j}{j+d(u)}\}\subseteq \{v_i\colon i\in\interf{0}{\uplim{u}}\}\subseteq \{W\}$ and b) by IH $\{W_i\lbr 0\rbr \}\subseteq \{W_i\}$ and $\{W_i\}\subseteq \{W\}$, we are done. \\
    \item Initially, it is true for $\satW(W'\lbr 0\rbr )$ by 1) above. It remains true for the subsequent calls $\satW(W_j\lbr 0\rbr ,\cdots)$ since by 1)  $\{W_j\lbr 0\rbr \}\subseteq \{W_j\}$ and since $\{W_j\}\subseteq \{W\}$, we have by IH $\{W\}\subseteq \{W'\}$. It also remains true for the calls $\satW(W\lbr j+1\rbr ,\cdots)$ since as a partial window $\{W\lbr j+1\rbr \}\subseteq \{W\}$ and we conclude again by IH. 
    \end{enumerate}
\end{proof}

\begin{lemma}[Soundness]\label{soundness}
If $u'$ is a $\Log2$-satisfiable \CCS\ then the call $\sat(u')$ returns \true. 
\end{lemma}

\begin{proof} 
Since $u'$ is $\Log2$-satisfiable, then $u'\neq \{\bot\}$. Hence the result of $\sat(u')$ rely on that of: \\
$\all \{\satW(\chooseW(u',\chooseCCS(\{\neg\phi\}\cup\Bom(u')),d(u')), u',d(u'),\uplim{u'})$\\
\indent $\colon\neg\Box\in u'\}$\\
We proceed by induction on $d(u')$\\
1) Case $d(u')=0$: then the set\\
    $\{\satW(\chooseW(u',\chooseCCS(\{\neg\phi\}\cup\Bom(u')),d(u')), u',d(u'),\uplim{u'})...\}$\\
    is empty. Hence $\sat(u')$ returns \true.\\
2) Case $d(u')\geq 1$. The induction hypothesis is IH$_1$: if $u$ is $\Log2$-satisfiable and $d(u)<d(u')$ then $\satW(u)$ returns \true. Now, for each $\neg\Box\phi\in u'$:\\
    2.1) if $k=0$ then by Corollary \ref{existence-of-window}, there exists $v'_0\in\CCS(\{\neg\phi\}\cup\Bom(u'))$ and there exists $W'$ a $(0,\chi(u'),\chi)$-window for $(u',v'_0)$ with all its $\CCS$ $\Log2$-satisfiable, namely $W'=\voidpair$.
    Thus by IH$_1$ (since $d(v'_0)<d(u')$), there exists $v'_0 \in \CCS(\{\neg\phi\}\cup \Bom(u'))$ such that $\sat(v'_0)$ returns \true.\\
    2.2) if $k>0$. By Corollary \ref{existence-of-window}, there exists $v'_0\in\CCS(\{\neg\phi\}\cup\Bom(u'))$ and there exists $W'$ a $(k',\chi(u'),\chi)$-window for $(u',v'_0)$ with all its $\CCS$ $\Log2$-satisfiable. We set:
    \begin{itemize}
        \item $\chooseCCS(\{\neg\phi\}\cup\Bom(u'))=v_0$
        \item $\chooseW(u',v'_0,k)=W'\lbr 0\rbr $
        \item and for each subsequent call $\satW(W\lbr 0\rbr ,u,k,\uplim{u})$ \\
        let $\nextW(W\lbr j\rbr ,u,k)=W\lbr j+1\rbr $ for $j\in\intero{0}{\chi(u)}$\\ \hfill(its $k$-continuation)
    \end{itemize}
Given that the initial call $\satW(W'\lbr 0\rbr ,u',k',\chi(u'))$, and that all \CCS\ of $W'$ are $\Log2$-satisfiable, then all calls $\satW(W\lbr j\rbr ,u,k,N)$ return \true. This can be proved by the following nested induction on $(k,N)$:\\
\begin{itemize}
        \item if $k=0$ (and $W=\voidpair$) or $N=0$ it is true since $\satW(W,\cdots)=\true$. 
        \item else ($k>0$ and $N>0$) with IH$_2$: $\satW(W\lbr j\rbr ,u,k,N)$ return \true; \\
        then: 
    \end{itemize}
    $\satW(W\lbr j\rbr ,u,k,N)$
            \begin{equation*}
        \begin{array}{ll}
        =&  \sat(v_0) \\ 
         &  \algand\,\,\satW(W_j\lbr 0\rbr ,v_{j+1},k-1,N)\\
         & \algand\,\,\satW(W\lbr j+1\rbr ,u,k,N-1)\\
        =&  \true \mbox{ (since by Prop. \ref{proposition-on-members}, $v_0\in \{W'\}$ and, as such,}\\
         &  \mbox{is satisfiable, thus $\sat(v_0)$ returns $\true$ by IH$_1$)}\\ 
         &  \algand\,\,\true \mbox{(by IH$_2$ since $(k-1)+N<k+N$)}\\
         & \algand\,\,\true \mbox{(by IH$_2$ since $k+(N-1)<k+N$)} 
        \end{array} 
        \end{equation*}
        In particular, $\satW(W'\lbr 0\rbr ,u',k',\uplim{u'})$ returns \true, and so does \\
     $\satW(\chooseW(u',\chooseCCS(\{\neg\phi\}\cup\Bom(u')),k'), u',k',\uplim{u'})$. \\ Consequently, $\sat(u')$ returns \true\ too.
  \end{proof}

For proving the completeness of this algorithm, we need to transform a \true\ into a model. To this aim we define the notion of satisfiability of a window. 
\paragraph{Satisfiability of window}
Let $M=(S,R,V)$ be a $2$-dense model and $x\in S$. Let $W$ be a $(k,n,\lambda)$-window for $(u,v_0)$. We say that $M$ satisfies $W$ at $x$, denoted by $M,x\models W$ iff:
\begin{itemize}
    \item $W=\voidpair$
    \item or, if $W=\win{v}{W}{n}$
    \begin{itemize}
        \item $M,x\models u$
        \item $\exists y_0\in R(x)\colon M,y_0\models v_0$
        \item if $i\in\intero{0}{n}$: 
        $\exists y_{i+1}\in R^-(y_{i})\cap R(u)\colon M,y_{i+1}\models v_{i+1}$ and $M,y_{i+1}\models W_{i}$
    \end{itemize}
\end{itemize}

\begin{lemma}[Completeness]\label{completeness}
Given a \CCS\ $u$ and a $(k,d(u),d)$-window for $(u,v_0)$, then:
\begin{itemize}
    \item [$\bullet$] if $\sat(u)$ returns \true\ then $u$ is $\Log2$-satisfiable
    \item [$\bullet$] if $d(u)\neq 0$ and $\satW(W,u,k,N)$ returns \true\ then $W$ is $\Log2$-satisfiable
\end{itemize}
\end{lemma}
\begin{proof}
    We construct $M=(S,R,V)$ by induction on $d(u)$. Let $\neg\Box\phi_0,\cdots,\neg\Box\phi_{n}$ be the $\Diamond$-formulas of $u$. In what follows we define $V_x$ by: $V_x(p)=\{x\}$ if $p\in x$ and else $V_x(p)=\emptyset$, for all $p\in\At$. \\
    If $d(u)=0$ the model $M=(\{u\}, \voidsequence,V_u)\models u$. Else, 
    \begin{itemize}
        \item Firstly, for each $\neg\Box\phi_l\in u$ suppose that
        $\sat(\chooseCCS(\{\neg\phi_l\}\cup \Bom(u)))$ returns \true. By IH, $\{\neg\phi_l\}\cup \Bom(u)$ is true in some $\Log2$-model: 
        $$(S^{l},R^{l},V^{l}),v^{l}\models \{\neg\phi_l\}\cup \Bom(u)$$
        \item Secondly, for each $\neg\Bom\phi_l\in s$ suppose that the call
        $$\satW(\chooseW(u,\chooseCCS(\{\neg\phi_l\}\cup\Bom(u)),k), u,k,\uplim{u})$$
        returns \true. Let $v_0$ be the \CCS\ chosen by $\chooseCCS$, and let $W^0=\win{v^0}{W}{u}$ the $(k,d(u),d)$-window for $u$ chosen by $\chooseW$. 
        Let $W^{j}=\win{v^j}{W^j}{u}$ (for $i\in\interf{1}{\uplim{u}}$), be the $(k,d(u),d)$-windows for $u$ chosen by the successive recursive calls to $\nextW$ (which succeed by hypothesis). Since each $W^{j+1}$ is a $k$-continuation of $W^j$, by repeated application of lemma \ref{k2kplusone-window}, with $(v_i,W_i)=(v^i_0,W^i_0)$ we obtain $\win{v}{W}{\uplim{u}}$ as a $(k,\uplim{u},d)$-window for $(u,v_0)$. Now, by applying lemma \ref{prolongation-to-infinite-window}, we can extend it to a $(k,\infty,d)$-window for $(u,v_0)$: 
        $W=\infwin{v}{W}$ where, beyond $\uplim{u}$, all $v_i$ and $W_i$ are copies of a $v_j$ and a $W_j$ with $j\leq \uplim{u}$. Since by hypothesis, for $i\geq 0$ all calls $\sat(v_i)$ and $\satW(W_i,v_{i+1},k+1,\chi(v_{i+1}))$ returns \true, then by induction hypothesis $v_i$ and $W_i$ are $\Log2$-satisfiable, let $M^0,x^0\models v_0$ and $M^{i+1},x^{i+1}\models W_i$ (this implies $M^{i+1},x^{i+1}\models v_{i+1}$ by definition).  \\
        In fact since they all depend on the formula $\neg\Box\phi_l$ involved, we add $l$ in the superscript giving: $W^{l}$, $v^{l}_i$, $W^{l}_i$, instead of just $W$, $v_i$ and $W_i$ and we write $M^{l,i},y^{l,i}\models W^{l}_{i-1}$, and $M^{l,i},y^{l,i}\models v_i^{l}$,  with $M^{l,i}=(S^{l,i},R^{l,i},V^{l,i})$. We merge these models into one, for each $\neg\Box\phi_l$ formula of $u$: $M^{l}=\bigsqcup_{i\geq 0} M^{l,i}$. \\
        Putting all things together, we define:
        \[M'=(S',R',V')=\bigsqcup_{l\in[1:n],\neg\Box\phi_l\in u} M^{l}
        \]
        \end{itemize}
        $M'$ is a $2$-dense model since it is the disjoint union of $2$-dense models. It remains to connect it with $u$ seen as a possible world to form the final model $M=(S,R,V)$: 
        \begin{itemize}
            \item $S=S'\sqcup \{u\}$
            \item $R''=R'\bigsqcup_{l\in[0:n],\neg\Box\phi_l\in u,i\geq  0}\{(u,y^{l,i})\}$
            \item $R=R''\sqcup\bigsqcup_{l\in[0:n_k],\neg\Box\phi_l\in u,i\geq 0}\{(y^{l,i+1},y^{l,i})\}$
            \item for each $p\in\At$: $V(p)=V'(p)\sqcup V_{u}(p)$
            \end{itemize}
        Now, it is time to check that 1. $M$ is $2$-dense, and (truth lemma) both 2. $M,u\models u$ and 3. $M,u\models  W^{k,l}$ are true:
    \begin{enumerate}
        \item Let $(x,y)\in R$: 
        \begin{itemize}
            \item if $(x,y)\in R''$, i.e.\ $(x,y)=(u,y^{l,i})$ then since $(u,y^{l,i+1})\in R$ and $(y^{l,i+1},y^{l,i})\in R$, we are done;
            \item if $(x,y)\in R'$ and $\not\in R''$, i.e.\ $(x,y)=(y^{l,i+1},y^{l,i})$, then:
            \begin{itemize}
                \item if $k=1$ then we are done;
                \item else since $W_i^{l}$ is $\Log2$-satisfiable, $\exists z\colon (y^{l,i+1},z)\in R$ and $(y^{l,i},z)\in R$.
            \end{itemize}
        \end{itemize}
        \item We only treat the case of modal formulas. \\
        $\Diamond$-formulas: Let $\neg\Box\neg\phi_l\in u$, since $\neg\phi_l\in v^{l}_0$ and $M,y^{l,0}\models v^{l}_0$, and  $(u,y^{l,0})\in R$, we are done;\\
        $\Box$-formulas: Let $\Box\neg\phi_l\in u$ and let $(u,y^{l,i})\in R$ for some $i$, $\neg\phi_l\in\Bom(u)$ hence $\neg\phi_l\in v_i^{l}$ which is a member of $W^{l}$, the $(k,\infty,\infty)$-window for $(u,v_0^{l})$ defined above. But since we added edges between worlds $y_i^{l}$ of distinct models, we must check that still $M,y_i^{l}\models v_i^{l}$ since $v_i^{l}$ may contain a $\Box$-formula that would be unsatisfied in $M$. We do so by a short induction on $i$; this is true for $M,y_{0}^{l}\models v_0^{l}$ (since no edge from $v_0^{l}$ were added); by definition of windows, we have that $\Bom v_{i+1}^{l}\subseteq v_i^{l}$ and since $M,y_i^{l}\models v_i^{l}$ (by IH on $i$) hence $M,y_{i+1}^{l}\models v_{i+1}^{l}$ and we are done again. 
        \item 
        \begin{itemize}
            \item [a)] $M,u\models u$ 
        \item [b)] $ y_0^{l}\in R(u)\colon M,y_0^{l}\models v_0^{l}$, 
        \item [c)] let $i\geq 0\colon y_{i+1}^{l}\in R^-(y_i^{l})\cap R(u)$; as seen just above $M,y_{i+1}^{l}\models v_{i+1}^{l}$ and, last, $M,y_{i+1}^{l}\models W_i^{l}$ (main induction hypothesis).  
        \end{itemize}
    \end{enumerate}
    All in all, $M,u\models u$ and  for all $\neg\Box\phi_l\in u$ $M,u\models W^{l}$
\end{proof}

\begin{lemma}
    The space needed for the algorithm is ${\cal O}(\lgu^3.2^{d(u)^5}))$. 
\end{lemma}
\begin{proof}
    First, we recall that functions \all\ and \algand\ are lazily evaluated. \\
    Obviously, \chooseCCS\ runs in polynomial space. 
    On another hand, as seen in lemma \ref{prolongation-to-infinite-window}, the size of each $(k,d(u),d)$-window for $u$ is bounded by $P(\lgu)$. Thus the functions \chooseW\ and \nextW\ run in polynomial space. It is also clear that functions $\sat$ and $\satW$ terminate since their recursion depth is bounded (respectively by $\lgu$ and by $N$) as well as their recursion width.
    Among all of these calls, let $w^{d(u)}$ be the argument with modal depth $d(u)-1$ for which $\sat$ has the maximum cost in terms of space, i.e.\ such that $\spc(\sat(w^{d(u)})$ is maximal. As well, among subwindows of $W$, let $W^{d(u)}$ be the $(k,d(v^{d(u)}),d)$-window, with $d(v^{d(u)})=d(u)-1$, for which the space used by $\satW(W^{d(u)},v^{d(u)},k,\uplim{v^{d(u)}})$ is maximal. \\
    Let $W^0=\win{v}{W^0}{d(u)}$ be the $(k,d(u),d)$-window for $(u,v_0)$ chosen by $\chooseW$, and for $i\in\interf{1}{\uplim{u}}$, let $W^{j}$ be the $(k,d(u),d)$-windows for $(u,v_i)$ chosen by the successive recursive calls to $\nextW$. \\
    Let us firstly evaluate the memory cost of $\satW(W^0,u,k,N)$, denoted by $\spc(\satW(W^0,u,k,N)))$. The function $\satW$  keeps its arguments in memory during the calls $\sat(v_0)$ and $\satW(W^0,u,k,N)$, then forget them and continue with $\satW(W^1,u,k,N-1)$. Let $\tau=\lgx{W^0}+\lgu+\lgx{k}+\lgx{\uplim{u}}$, we have $\tau\leq 4.\lgx{W^+}$, and the following inequalities:\\ 
    $\spc(\satW(W^0,u,k,\uplim{u}))$\\
\begin{minipage}{\columnwidth}
\begin{supertabular}{lll}
    $\leq$& $\max \{$&$\tau+\spc(\sat(v_0)),$\\
    &&$\tau+\spc(\satW(W_i^0,v_i^0,k-1,\uplim{v_i^0})),$\\
    &&$space(\satW(W^1,u,k,\uplim{u}-1))\}$\\
    $\leq$&$ \max \{$&$\tau+\spc(\sat(v^{d(u)})),$\\
    & & $\tau+\spc(\satW(W^{d(u)},v^{d(u)},k-1,\uplim{v^{d(u)}})),$\\
    & & $space(\satW(W^1,u,k,\uplim{u}-1))\}$\\
    $\leq$ & $ \max \{$ & $\tau+\spc(\sat(v^{d(u)})),$\\
    & & $\tau+\spc(\satW(W^{d(u)},v^{d(u)},k-1,\uplim{v^{d(u)}})),$\\
    & &$\max\{\tau+\spc(\sat(v^{d(u)})),$\\
    & & \hspace{0.7cm}$\tau+\spc(\satW(W^{d(u)},v^{d(u)},k-1,\uplim{v^+})),$\\ 
        & & $\hspace{0.7cm}\space(\satW(W^2,u,k,\uplim{u}-2))\}\}$\\
    $\leq$ & $\max \{ $ & $\tau+\spc(\sat(v^{d(u)})),$\\
    & & $\tau+\spc(\satW(W^{d(u)},v^{d(u)},k-1,\uplim{v^{d(u)}})),$\\
    & & $\spc(\satW(W^2,u,k,\uplim{u}-2))\}$\\
    $\leq $ &  $ \max \{ $ & $ \tau+\spc(\sat(v^{d(u)})) ,$\\
 & & $ \tau+\spc(\satW(W^{d(u)},v^{d(u)},k-1,\uplim{v^{d(u)}})), $\\
     & & $  \hspace{1.2cm}\spc(\satW(W^{\uplim{u}},u,k,0))\} $ \\
     $ \leq $ & & $ \tau+\max \{\spc(\sat(v^{d(u)})),$\\
    & & $ \hspace{1.2cm}\spc(\satW(W^{d(u)},v^{d(u)},k-1,\uplim{v^{d(u)}}))\}$\\
      $ \leq $ & & $ \tau+\max \{\spc(\sat(v^{d(u)})),$\\
    & &  \hspace{1.2cm} $\tau+\max \{\spc(\sat(v^{d(u)})),$\\
    & & \hspace{1.2cm} $\spc(\satW(W^{d(u)},v^{d(u)},k-2,\uplim{v^{d(u)}}))\}\}$\\

$ \leq $ & & $ 2.\tau+\max \{\spc(\sat(v^{d(u)})),$\\
    & & \hspace{1.3cm} $\spc(\satW(W^{d(u)},v^{d(u)},k-2,\uplim{v^{d(u)}}))\}$\\
      $ \leq $ & & $ k.\tau+\max \{\spc(\sat(v^{d(u)})),$\\
    & & \hspace{1.3cm} $\spc(\satW(W^+,v^+,0,\uplim{v^{d(u)}}))\}$\\
       $ \leq $ & & $ k.\tau+\spc(\sat(v^{d(u)}))$\\
       &&\\
\end{supertabular}
\end{minipage}

Now, concerning the function $\sat$, it also keeps track of its argument in memory during recursion in order to range over its $\Diamond$-formulas. obviously, in general, calls $\satW$ need more space than $\sat$ calls. For $\neg\Box\phi\in u$, let $W^{0,\neg\Box\phi}$ be the $(k,d(u),d)$-window chosen by $\chooseW(u,\chooseCCS(\{\neg\phi\}\cup\Bom(u)),k), u,k,\uplim{u})$ (it exists, otherwise the algorithm stops). 
    Thus:\\
    $\spc(\sat(u))$\\
    $\begin{array}{ll}
     &  \leq \lgu+\max \{\spc(\satW(W^{0,\neg\Box\phi},u,k,\uplim{u}))\colon\neg\Box\phi\in u\}\\
    & \leq \lgu+k.\tau+\spc(\sat(v^{d(u)-1}))\\
    &\leq 2.\tau^2+\spc(\sat(v^{d(u)-1}))\leq\cdots\leq 2.\tau^2.d(u)\leq 2.\tau^3\\
    &\leq c.\lgx{W^{d(u)}}^3 \mbox{ for some constant $c>0$}
    \end{array}$
    
From lemma \ref{prolongation-to-infinite-window}, we know that $\card({W^{d(u)}})=Q(d(u))$ (for a polynomial $Q$ of degree $d(u)+1$), with elements of size bounded by $c_{\SF(u)}.\lgu$, hence $\lgx{W^{d(u)}}=c_{\SF(u)}.\lgu.Q(d(u))$ and $\spc(\sat(u))\leq c'.\lgu^3.d(u)^{d(u)^3+3}\leq c'.\lgu^3.2^{d(u)^5}$ for some constant $c'>0$.

\end{proof}

\begin{corollary}
    $\Log2$-satisfiability is in para-$\PSPACE$.
\end{corollary}
\begin{proof}
   If we go back to the definition of the para-$\PSPACE$ class and take $\kappa(u)=d(u)$, the function $f(\kappa(u))=2^{d(u)^5}$ is clearly computable and $c'.\lgu^3$ is indeed a polynomial. Hence $2$-satisfiability is in para-$\PSPACE$.  
\end{proof}

\section{Generalization to \emph{n}-dense logics}\label{general_case}
For handling the general case (say $\Box^m p\rightarrow\Box p$) we must take into account that between $u$ and $u_0$ we must introduce (at the first level) nodes $u_1$ up to $u_{m-1}$ (as well as the necessary subwindows). But now, a $n$-long window (of $n+1$ nodes and $n$ subwindows) will rather contain $m$-uples of nodes and $(m-1)$-uples of windows, $n-1$ of each.

Formally, Let $u$, $v_0$ be two \CCS, $k\leq d(u)$ and $n\geq d(u)$ and suppose we define windows for $\Log{m}$. 
A $(k,n,\lambda)$-window for $(u,v_0)$ denoted by $W$ is a pair $\pair{\mathcal N}{\mathcal W}$ (with $I=\interf{0}{\sizew}$, $I^-=\intero{0}{\sizew}$):\\
$\bullet$ if $k=0$, $\pair{\mathcal N}{\cal W}=\voidpair$ (empty window)\\
$\bullet$ if $k>0$, $\pair{\mathcal N}{\cal W}$ is a pair of two sequences: 
        \begin{itemize}
            \item [-]${\mathcal N}=(v_i^0,\cdots,v_i^{m-1})_{i\in I}$ is a sequence of $m$-uples of nodes s.th.:
        \begin{itemize}   
            \item [*] for $i<n$: $u_i^{m-1}=u_{i+1}^0$
            \item [*]$\forall i\in\intero{0}{\sizew}\forall 0<j<m-1\colon $\\
            $v_i^{m-1}\in\CCS(\Bom(u)\cup\Bom(v_{i+1}^0))$\\
            and if $i<n$: $v^j_i\in\CCS(\Bom(v^{j+1}_i))$,  \\
             and if $\sizew\neq \infty$ then and $v^j_n\in\CCS(\Bom(v^{j+1}_n))$
        \end{itemize}
        \item [-] ${\cal W}=(W_i^0,\cdots,W_i^{m-2})_{i\in I}$ is a sequence of $m$-uples of $(k-1,\lambda(v_{i+1}),\lambda)$-windows for $(v_i^j,v_i^{j+1})$. 
        \end{itemize}
This will look as in Fig. \ref{part_model5}. 
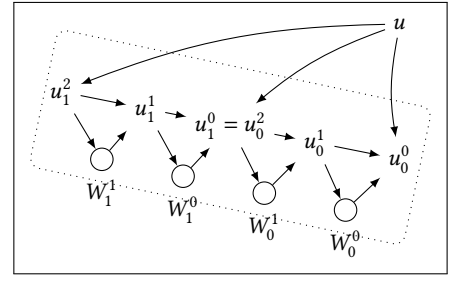
\begin{figure}[t]
\begin{centering}
  \fbox{
  \begin{tikzpicture}[scale=0.9]
\node (u) at (2,1) {$u$};
\node (u0) at (2,-1) {$u_0^0$};
\node (u1) at (0.75,-0.75) {$u_0^1$};
\node (u2) at (-0.5,-0.5) {$u_1^0=u_0^2$};
\node (u3) at (-1.75,-0.25) {$u_1^1$};
\node (u4) at (-3,0) {$u_1^2$};

\node[draw,circle] (w0) at (1.2,-1.75) {};
\node[draw,circle] (w1) at (0,-1.5) {};
\node[draw,circle] (w2) at (-1.2,-1.25) {};
\node[draw,circle] (w3) at (-2.4,-1) {};

\draw [->,>=latex] (u2) -- (u1);
\draw [->,>=latex] (u4) -- (u3);
\draw [->,>=latex] (u1) -- (u0);
\draw [->,>=latex] (u3) -- (u2);
\draw [->,>=latex] (u1) -- (w0);
\draw [->,>=latex] (w0) -- (u0);
\draw [->,>=latex] (u2) -- (w1);
\draw [->,>=latex] (w1) -- (u1);
\draw [->,>=latex] (u3) -- (w2);
\draw [->,>=latex] (w2) -- (u2);
\draw [->,>=latex] (u4) -- (w3);
\draw [->,>=latex] (w3) -- (u3);
\draw (w0.south) node[below]{$W^0_0$};
\draw (w1.south) node[below]{$W^1_0$};
\draw (w2.south) node[below]{$W^0_1$};
\draw (w3.south) node[below]{$W^1_1$};
\foreach \n in {u0,u2,u4} \draw [->,>=latex] (u)  to[bend right=10](\n);
\node[draw, dotted, rounded corners,rotate=-12, fit=(u0)(u1)(u2) (u3)(u4) (w1)(w2) (w3), inner sep=1pt, label=below:{\small }] {};
\end{tikzpicture}
}
\caption{A $2$-long window for $3$-density}\label{part_model5}
\Description{A $n$-dense window and its continuation}
\end{centering}
\end{figure} 

Definition of continuations must be modified so as the shift is now made of $m-1$ nodes from a window to its continuation and correspond to Fig.\ \ref{continuation2}. 
\begin{figure}[h]
\begin{centering}
  \fbox{
  \begin{tikzpicture}[scale=0.7]
\node (u) at (2,1) {$u$};
\node (u0) at (2,-1) {$u_0^0$};
\node (u1) at (0.75,-0.75) {$u_0^1$};
\node (u2) at (-0.5,-0.5) {$u_1^0$};
\node (u3) at (-1.75,-0.25) {$u_1^1$};
\node (u4) at (-3,0) {$u_1^2$};
\node (u5) at (-4.25,0.25) {$u_2^0$};
\node (u6) at (-5.5,0.5) {$u_2^1$};

\node[draw,circle] (w0) at (1.2,-1.75) {};
\node[draw,circle] (w1) at (-0.1,-1.5) {};
\node[draw,circle] (w2) at (-1.3,-1.25) {};
\node[draw,circle] (w3) at (-2.55,-1) {};
\node[draw,circle] (w4) at (-3.85,-0.75) {};
\node[draw,circle] (w5) at (-5.1,-0.5) {};

\draw [->,>=latex] (u2) -- (u1);
\draw [->,>=latex] (u4) -- (u3);
\draw [->,>=latex] (u1) -- (u0);
\draw [->,>=latex] (u3) -- (u2);
\draw [->,>=latex] (u6) -- (u5);
\draw [->,>=latex] (u5) -- (u4);
\draw [->,>=latex] (u1) -- (w0);
\draw [->,>=latex] (w0) -- (u0);
\draw [->,>=latex] (u2) -- (w1);
\draw [->,>=latex] (w1) -- (u1);
\draw [->,>=latex] (u3) -- (w2);
\draw [->,>=latex] (w2) -- (u2);
\draw [->,>=latex] (u4) -- (w3);
\draw [->,>=latex] (w3) -- (u3);
\draw [->,>=latex] (u6) -- (w5);
\draw [->,>=latex] (w5) -- (u5);
\draw [->,>=latex] (u5) -- (w4);
\draw [->,>=latex] (w4) -- (u4);
\foreach \n in {u0,u2,u4,u6} \draw [->,>=latex] (u)  to[bend right=10](\n);
\node[draw, dotted, rounded corners,rotate=-12, fit=(u0)(u1)(u2) (u3)(u4) (w1)(w2) (w3), inner sep=1pt, label=below:{\small }] {};
\node[draw, dashed, rounded corners,rotate=-12, fit=(u2)(u3)(u4)(u5)(u6)  (w3)(w4)(w5), inner sep=-1pt, label=below:{\small }] {};
\end{tikzpicture}
}
\caption{A $2$-long window for $3$-density (dotted part) and its continuation (dashed part)}\label{continuation2}
\Description{A $n$-dense window and its continuation}
\end{centering}
\end{figure}
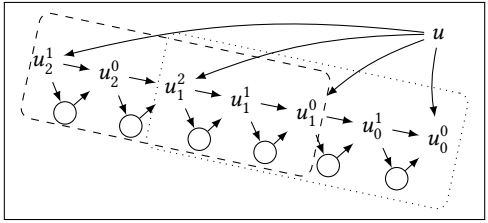 

All of lemmas \ref{k2kplusone-window}, \ref{prolongation-to-infinite-window}, \ref{from-model-to-window}, \ref{soundness} and \ref{completeness} transfer tediously but with no difficulty and since windows contain $m.d(u)$ nodes instead of $d(u)$, then $\uplim{u}$ becomes $2^{P(m.\lgu)}+m.d(u)$, and the size of $(k,d(u),d)$-windows is now $Q(m.d(u))$ for some polynomial $Q$ of degree $k$. This yields $\spc(\sat(u))\leq c'.\lgu^3.(m.d(u))^{d(u)^3+3}\leq c'.\lgu^3.2^{2.d(u)^5}$ for some constant $c'>0$ since $m$ is a constant. 

\begin{corollary}
    For all $n\in\nat\colon\Log{n}$-satisfiability is in para-$\PSPACE$.
\end{corollary}

\section*{Conclusion}\label{conclusion}
Interestingly, the recent concept of windows seems to be useful apart from its initial aim of browsing a model/tableau in polynomial space. After having been applied, in a non recursive style to weak-density, it proves to be useful in a quite different case. In this paper, with the non-trivial generalization to \emph{recursive windows}, we could design an algorithms for $n$-dense logics which runs in polynomial space \emph{for fixed modal depth of the input}, and thus assesses the membership of the $n$-dense satisfiability problem in the para-$\PSPACE$ class. Hopefully, windows could be applicable to other open questions of complexity or decidability for logics having properties involving existence of intermediary worlds in frames --for example extension of $n$-dense logics to the multimodal case-- by defining specific window structures. In any case, we believe parameterized complexity has been globally understudied in the field of modal logic and deserves to be explored more at least for all these logics of rather high complexity.


\bibliographystyle{ACM-Reference-Format}
\bibliography{para-n-dense}
\end{document}